\documentclass[envcountsame,runningheads]{llncs}
\usepackage{amsmath}

\newcommand{\cmin}{\bar{c}}
\newcommand{\rot}{\textsc{rot}}
\newcommand{\suf}[2]{#1[#2 .. |#1| - 1]}
\newcounter{instr}
\newcommand{\ninstr}{\refstepcounter{instr}\theinstr.}

\title{Alternative Algorithms for Lyndon Factorization\thanks{Supported by the Academy of Finland (grant 134287).}}

\author{Sukhpal Singh Ghuman\inst{1}\and Emanuele Giaquinta\inst{2}\and Jorma Tarhio\inst{1}}
\institute{
 Department of Computer Science and Engineering, Aalto University\\P.O.B.\ 15400, FI-00076 Aalto, Finland\\
  \mbox{ \email{\{Sukhpal.Ghuman,Jorma.Tarhio\}@aalto.fi}} \\ \and
Department of Computer Science, P.O.B.\ 68, FI-00014  University of Helsinki, Finland \\
  \email{Emanuele.Giaquinta@cs.helsinki.fi} }

\begin{document}

\maketitle

\begin{abstract}
  We present two variations of Duval's algorithm for computing the Lyndon factorization of a word. The first algorithm is designed for the case of small alphabets and is able to skip a significant portion of the characters of the string, for strings containing runs of the smallest character in the alphabet. Experimental results show that it is faster than Duval's original algorithm, more than ten times in the  case of long DNA strings. The second algorithm computes, given a run-length encoded string $R$ of length $\rho$, the Lyndon factorization of $R$ in $O(\rho)$ time and constant space.
\end{abstract}

\section{Introduction}

Given two strings $w$ and $w'$, $w'$ is a rotation of $w$ if $w = uv$
and $w' = vu$, for some strings $u$ and $v$. A string is a Lyndon word
if it is lexicographically smaller than all its proper rotations.
Every string has a unique factorization in Lyndon words such that the
corresponding sequence of factors is nonincreasing with respect to
lexicographical order. This factorization was introduced by Chen, Fox
and Lyndon~\cite{ChenFL58}. Duval's classical algorithm \cite{Duval83}
computes the factorization in linear time and constant space. The
Lyndon factorization is a key ingredient in a recent method for
sorting the suffixes of a text~\cite{MantaciRRS13}, which is a
fundamental step in the construction of the Burrows-Wheeler transform
and of the suffix array, as well as in the bijective variant of the
Burrows-Wheeler
transform~\cite{DBLP:journals/corr/abs-1201-3077}~\cite{Kufleitner09}.
The Burrows-Wheeler transform is an invertible transformation of a
string, based on the sorting of its rotations, while the suffix array
is a lexicographically sorted array of the suffixes of a string. They
are the basis for important data compression methods and text indexes.
Although Duval's algorithm runs in linear time and is thus efficient, it
can still be useful to further improve the time for the computation of
the Lyndon factorization in the cases where the string is either huge
or compressible and given in a compressed form.

Various alternative algorithms for the Lyndon factorization have been
proposed in the last twenty years. Apostolico and
Crochemore presented a parallel algorithm~\cite{ApostolicoC95}, while
Roh \emph{et al.} described an external memory algorithm~\cite{RohCIP08}.
Recently, I \emph{et al.} showed how to compute the Lyndon factorization of a
string given in grammar-compressed form and in Lempel-Ziv 78
encoding~\cite{spire/INIBT13}.

In this paper, we present two variations of Duval's algorithm. The
first variation is designed for the case of small alphabets like the
DNA alphabet \{a, c, g, t\}. If the string contains runs of the
smallest character, the algorithm is able to skip a significant
portion of the characters of the string. In our experiments, the new
algorithm is more than ten times faster than the original one for long DNA
strings.

The second variation is for strings compressed with run-length
encoding. The run-length encoding of a string is a simple encoding
where each maximal consecutive sequence of the same symbol is encoded as
a pair consisting of the symbol plus the length of the sequence. Given
a run-length encoded string $R$ of length $\rho$, our algorithm
computes the Lyndon factorization of $R$ in $O(\rho)$ time and uses
constant space. It is thus preferable to Duval's algorithm in the
cases in which the strings are stored or maintained in run-length
encoding.

\section{Basic definitions}

Let $\Sigma$ be a finite ordered alphabet of symbols and let
$\Sigma^*$ be the set of words (strings) over $\Sigma$ ordered by
lexicographic order. The empty word $\varepsilon$ is a word of length
$0$. Let also $\Sigma^+$ be equal to $\Sigma^*\setminus \{ \varepsilon
\}$. Given a word $w$, we denote with $|w|$ the length of $w$ and with
$w[i]$ the $i$-th symbol of $w$, for $0\le i < |w|$. The concatenation
of two words $u$ and $v$ is denoted by $uv$. Given two words $u$ and
$v$, $v$ is a substring of $u$ if there are indices $0\le i,j < |u|$
such that $v = u[i] ... u[j]$. If $i = 0$ ($j = |u| - 1$) then $v$ is
a prefix (suffix) of $u$. We denote by $u[i .. j]$ the substring of
$u$ starting at position $i$ and ending at position $j$. For $i > j$
$u[i .. j] = \varepsilon$. We denote by $u^k$ the concatenation of $k$
$u$'s, for $u\in\Sigma^+$ and $k\ge 1$. The longest border of a word
$w$, denoted with $\beta(w)$, is the longest proper prefix of $w$ which is also a suffix of $w$.
Let $lcp(w, w')$ denote the length of the longest common prefix of words $w$ and $w'$.
We write $w < w'$ if either $lcp(w, w') = |w| < |w'|$, i.e., if $w$ is a proper prefix of $w'$, or if $w[lcp(w, w')] < w'[lcp(w, w')]$.
For any $0\le i < |w|$, $\rot(w, i) = w[i .. |w| - 1] w[0 .. i-1]$ is a rotation of $w$.
A Lyndon word is a word $w$ such that $w < \rot(w, i)$, for $1 \le i <
|w|$. Given a Lyndon word $w$, the following properties hold:
\begin{enumerate}
\item $|\beta(w)| = 0$;
\item either $|w| = 1$ or $w[0] < w[|w| - 1]$.
\end{enumerate}
Both properties imply that no word $a^k$, for $a\in\Sigma$,
$k\ge 2$, is a Lyndon word. The following result is due to Chen, Fox and Lyndon~\cite{cow}:
\begin{theorem}
  Any word $w$ admits a unique factorization $CFL(w) = w_1, w_2, \ldots, w_m$,
  such that $w_i$ is a Lyndon word, for $1\le i\le m$, and $w_1\ge
  w_2\ge \ldots \ge w_m$.
\end{theorem}
The run-length encoding (RLE) of a word $w$, denoted by $\textsc{rle}(w)$,
is a sequence of pairs (runs) $\langle (c_1, l_1), (c_2, l_2,), \ldots, (c_{\rho},
l_{\rho}) \rangle$ such that $c_i\in \Sigma$, $l_i\ge 1$, $c_i\neq c_{i+1}$ for
$1 \le i < r$, and $w = c_1^{l_1} c_2^{l_2} \ldots c_{\rho}^{l_{\rho}}$. The
interval of positions in $w$ of the factor $w_i$ in the Lyndon
factorization of $w$ is $[a_i, b_i]$, where
$a_i=\sum_{j=1}^{i-1}|w_j|$, $b_i=\sum_{j=1}^{i}|w_j| - 1$. Similarly, the
interval of positions in $w$ of the run $(c_i, l_i)$ is $[a^{rle}_i, b^{rle}_i]$
where $a^{rle}_i=\sum_{j=1}^{i-1} l_j$, $b^{rle}_i=\sum_{j=1}^i l_j - 1$.

\section{Duval's algorithm}

\setcounter{instr}{0}
\begin{figure}[t]
\begin{center}
\begin{tabular}{|rl|}
\hline
\multicolumn{2}{|l|}{\textsc{LF-Duval}$(w)$}\\
\ninstr & $k\leftarrow 0$ \\
\ninstr & \textbf{while} $k < |w|$ \textbf{do} \\
\ninstr & \qquad \qquad $i\leftarrow k + 1$ \\
\ninstr & \qquad \qquad $j\leftarrow k + 2$ \\
\ninstr & \qquad \qquad \textbf{while} \textsc{true} \textbf{do} \\
\ninstr & \qquad \qquad \qquad \textbf{if} $j = |w| + 1$ \textbf{or} $w[j-1] < w[i-1]$ \textbf{then} \\
\ninstr & \qquad \qquad \qquad \qquad \textbf{while} $k < i$ \textbf{do} \\
\ninstr & \qquad \qquad \qquad \qquad \qquad \textbf{output}($w[k .. k + j - i]$) \\
\ninstr & \qquad \qquad \qquad \qquad \qquad $k\leftarrow k + j - i$ \\
\ninstr & \qquad \qquad \qquad \qquad \textbf{break} \\
\ninstr & \qquad \qquad \qquad \textbf{else} \\
\ninstr & \qquad \qquad \qquad \qquad \textbf{if} $w[j-1] > w[i-1]$ \textbf{then} \\
\ninstr & \qquad \qquad \qquad \qquad \qquad $i\leftarrow k + 1$ \\
\ninstr & \qquad \qquad \qquad \qquad \textbf{else} \\
\ninstr & \qquad \qquad \qquad \qquad \qquad $i\leftarrow i + 1$ \\
\ninstr & \qquad \qquad \qquad \qquad $j\leftarrow j + 1$ \\
\hline
\end{tabular}
\end{center}
\caption{Duval's algorithm to compute the Lyndon factorization of a string.}
\label{fig:duval}
\end{figure}

In this section we briefly describe Duval's algorithm for the
computation of the Lyndon factorization of a word. Let $L$ be the set
of Lyndon words and let
$$
P = \{ w\ |\ w\in\Sigma^+ \text{ and } w\Sigma^*\cap L\neq \emptyset \}\,,
$$ be the set of nonempty prefixes of Lyndon words. Let also $P' =
P\cup \{ c^k\ |\ k\ge 2\}$, where $c$ is the maximum symbol in
$\Sigma$. Duval's algorithm is based on the following Lemmas, proved
in~\cite{Duval83}:
\begin{lemma}\label{lemma:lf-duval1}
  Let $w\in \Sigma^+$ and $w_1$ be the longest prefix of
  $w = w_1 w'$ which is in $L$. We have $CFL(w) = w_1 CFL(w')$.
\end{lemma}
\begin{lemma}\label{lemma:lf-duval2}
$P' = \{ (uv)^k u\ |\ u\in\Sigma^*, v\in\Sigma^+, k\ge 1 \text{ and } uv\in L \}$.
\end{lemma}
\begin{lemma}\label{lemma:lf-duval3}
Let $w = (uav')^ku$, with $u,v'\in\Sigma^*$, $a\in\Sigma$, $k\ge 1$ and $uav'\in L$. The following propositions hold:
\begin{enumerate}
\item For $a'\in\Sigma$ and $a > a'$, $wa'\notin P'$;
\item For $a'\in\Sigma$ and $a < a'$, $wa'\in L$;
\item For $a' = a$, $wa'\in P'\setminus L$.
\end{enumerate}
\end{lemma}
Lemma~\ref{lemma:lf-duval1} states that the computation of the Lyndon
factorization of a word $w$ can be carried out by computing the
longest prefix $w_1$ of $w = w_1 w'$ which is a Lyndon word and then
recursively restarting the process from $w'$. Lemma~\ref{lemma:lf-duval2} states
that the nonempty prefixes of Lyndon words are all of the form $(uv)^k
u$, where $u\in\Sigma^*, v\in\Sigma^+, k\ge 1 \text{ and } uv\in L$.
By the first property of Lyndon words, the longest prefix of $(uv)^k
u$ which is in $L$ is $uv$. Hence, if we know that $w = (uv)^k u a
v'$, $(uv)^k u\in P'$ but $(uv)^k u a\notin P'$, then by Lemma~\ref{lemma:lf-duval1} and
by induction we have $CFL(w) = w_1 w_2 \ldots w_k\, CFL(u a v')$,
where $w_1 = w_2 = \ldots = w_k = uv$. Finally, Lemma~\ref{lemma:lf-duval3}
explains how to compute, given a word $w\in P'$ and a symbol
$a\in\Sigma$, whether $w a\in P'$, and thus makes it possible to
compute the factorization using a left to right parsing. Note that,
given a word $w\in P'$ with $|\beta(w)| = i$, we have $w[0 ..
  |w|-i-1]\in L$ and $w = (w[0 .. |w|-i-1])^q w[0 .. r-1]$ with $q
=\lfloor \frac{|w|}{|w| - i}\rfloor$ and $r = |w|\mod (|w| - i)$. For
example, if $w = abbabbab$, we have $|w| = 8$, $|\beta(w)| = 5$, $q =
2$, $r = 2$ and $w = (abb)^2 ab$.
The code of Duval's algorithm is shown in Figure~\ref{fig:duval}.

The following is an alternative formulation of Duval's algorithm by I
\emph{et al.}~\cite{spire/INIBT13}:
\begin{lemma}\label{lemma:lf-lcp-skip}
  Let $j > 0$ be any position of a string $w$ such that $w <
  \suf{w}{i}$ for any $0 < i\le j$ and $lcp(w, \suf{w}{j})\ge 1$.
  Then, $w < \suf{w}{k}$ also holds for any $j < k\le j + lcp(w,
  \suf{w}{j})$.
\end{lemma}
\begin{lemma}\label{lemma:lf-lcp}
  Let $w$ be a string with $CFL(w) = w_1, w_2, \ldots, w_m$. It holds
  that $|w_1| = \min \{ j\ |\ \suf{w}{j} < w\}$ and $w_1 = w_2
  = \ldots = w_k = w[0 .. |w_1| - 1]$, where $k = 1 +
  \lfloor lcp(w, \suf{w}{|w_1|}) / |w_1|\rfloor$. 
\end{lemma}
Based on these Lemmas, Duval's algorithm can be implemented by
initializing $j\leftarrow 1$ and executing the following steps until
$w$ becomes $\varepsilon$: 1) compute $h\leftarrow lcp(w,
\suf{w}{j})$. 2) if $j + h < |w|$ and $w[h] < w[j + h]$ set $j\leftarrow j + h + 1$;
otherwise output $w[0 .. j - 1]$ $k$ times and set $w\leftarrow
\suf{w}{jk}$, where $k = 1 + \lfloor h / j\rfloor$, and set
$j\leftarrow 1$.

\section{Improved algorithm for small alphabets}

\setcounter{instr}{0}
\begin{figure}[t]
\begin{center}
\begin{tabular}{|rl|}
\hline
\multicolumn{2}{|l|}{\textsc{LF-skip}$(w)$}\\
\ninstr & $e\leftarrow |w| - 1$ \\
\ninstr & \textbf{while} $e\ge 0$ \textbf{and} $w[e] = \cmin$ \textbf{do} \\
\ninstr & \qquad $e\leftarrow e - 1$ \\
\ninstr & $l\leftarrow |w| - 1 - e$ \\
\ninstr & $w\leftarrow w[0 .. e]$ \\
\ninstr & $s\leftarrow \min Occ_{\{\cmin\cmin\}}(w)\cup \{ |w| \}$ \\
\ninstr & \textsc{LF-Duval}($w[0 .. s-1]$) \\
\ninstr & $r\leftarrow 0$ \\
\ninstr & \textbf{while} $s < |w|$ \textbf{do} \\
\ninstr & \qquad $w\leftarrow w[s .. |w| - 1]$ \\
\ninstr & \qquad \textbf{while} $w[r] = \cmin$ \textbf{do} \\
\ninstr & \qquad \qquad $r\leftarrow r + 1$ \\
\ninstr & \qquad $s\leftarrow |w|$ \\
\ninstr & \qquad $k\leftarrow 1$ \\
\ninstr & \qquad $\mathcal{P}\leftarrow \{ \cmin^r c\ |\ c\le w[r] \}$ \\
\ninstr & \qquad $j\leftarrow 0$ \\
\ninstr & \qquad \textbf{for} $i\in Occ_{\mathcal{P}}(w): i > j$ \textbf{do} \\
\ninstr & \qquad \qquad $h\leftarrow lcp(w, \suf{w}{i})$ \\
\ninstr & \qquad \qquad \textbf{if} $h = |w| - i$ \textbf{or} $w[i + h] < w[h]$ \textbf{then} \\
\ninstr & \qquad \qquad \qquad $s\leftarrow i$ \\
\ninstr & \qquad \qquad \qquad $k\leftarrow 1 + \lfloor h / s\rfloor$ \\
\ninstr & \qquad \qquad \qquad \textbf{break} \\
\ninstr & \qquad \qquad $j\leftarrow i + h$ \\
\ninstr & \qquad \textbf{for} $i\leftarrow 1$ \textbf{to} $k$ \textbf{do} \\
\ninstr & \qquad \qquad \textbf{output}($w[0 .. s - 1]$) \\
\ninstr & \qquad $s\leftarrow s \times k$ \\
\ninstr & \textbf{for} $i\leftarrow 1$ \textbf{to} $l$ \textbf{do} \\
\ninstr & \qquad \textbf{output}($\cmin$) \\
\hline
\end{tabular}
\end{center}
\caption{The algorithm to compute the Lyndon factorization that can potentially skip symbols.}
\label{fig:duval-variant}
\end{figure}

Let $w$ be a word over an alphabet $\Sigma$ with $CFL(w) = w_1, w_2,
\ldots, w_m$ and let $\cmin$ be the smallest symbol in $\Sigma$.
Suppose that there exists $k\ge 2, i\ge 1$ such that $\cmin^k$ is a
prefix of $w_i$. If the last symbol of $w$ is not $\cmin$, then by
Theorem $1$ and by the properties of Lyndon words, $\cmin^k$ is a
prefix of each of $w_{i+1}, w_{i+1}, \ldots, w_{m}$. This property can
be exploited to devise an algorithm for Lyndon factorization that can
potentially skip symbols. Our algorithm is based on the alternative
formulation of Duval's algorithm by I \emph{et al.}. Given a set of
strings $\mathcal{P}$, let $Occ_{\mathcal{P}}(w)$ be the set of all
(starting) positions in $w$ corresponding to occurrences of the
strings in $\mathcal{P}$. We start with the following Lemmas:
\begin{lemma}\label{lemma:lf-tail}
  Let $w$ be a word and let $s = \max \{ i\ |\ w[i] > \cmin \}\cup \{
  - 1 \}$. Then, $CFL(w) = CFL(w[0 .. s]) CFL(\cmin^{(|w| - 1 -
    s)})$.
\end{lemma}
\begin{proof}
  If $s = -1$ or $s = |w| - 1$ the Lemma plainly holds. Otherwise, Let
  $w_i$ be the factor in $CFL(w)$ such that $s\in [a_i, b_i]$. To
  prove the claim we have to show that $b_i = s$. Suppose by
  contradiction that $s < b_i$, which implies $|w_i|\ge 2$. Then,
  $w_i[|w_i| - 1] = \cmin$, which contradicts the second property of
  Lyndon words. \qed
\end{proof}
\begin{lemma}\label{lemma:lf-split}
  Let $w$ be a word such that $\cmin\cmin$ occurs in it and let $s =
  \min Occ_{\{\cmin\cmin\}}(w)$. Then, we have $CFL(w) = CFL(w[0 .. s
  - 1]) CFL(w[s .. |w| - 1])$.
\end{lemma}
\begin{proof}
  Let $w_i$ be the factor in $CFL(w)$ such that $s\in [a_i, b_i]$. To
  prove the claim we have to show that $a_i = s$. Suppose by
  contradiction that $s > a_i$, which implies $|w_i|\ge 2$. If $s =
  b_i$ then $w_i[|w_i|-1] = \cmin$, which contradicts the second
  property of Lyndon words. Otherwise, since $w_i$ is a Lyndon word it
  must hold that $w_i < \rot(w_i, s - a_i)$. This implies at least
  that $w_i[0] = w_i[1] = \cmin$, which contradicts the hypothesis
  that $s$ is the smallest element in $Occ_{\{\cmin\cmin\}}(w)$.\qed
\end{proof}
\begin{lemma}\label{lemma:lf-skip}
  Let $w$ be a word such that $w[0] = w[1] = \cmin$ and $w[|w| -
  1]\neq \cmin$. Let $r$ be the smallest position in $w$ such that
  $w[r]\neq \cmin$. Note that $w[0 .. r - 1] = \cmin^r$. Let also
  $\mathcal{P} = \{ \cmin^r c\ |\ c \le w[r] \}$. Then we have
$$
b_1 = \min \{ s\in Occ_{\mathcal{P}}(w)\ |\ \suf{w}{s} < w \} \cup \{ |w| \} - 1 \,,
$$
where $b_1$ is the ending position of factor $w_1$.
\end{lemma}
\begin{proof}
  By Lemma~\ref{lemma:lf-lcp} we have that $b_1 = \min \{ s\ |\ w[s ..
  |w| - 1] < w\} - 1$. Since $w[0 .. r - 1] = \cmin^r$ and $|w|\ge r +
  1$, for any string $v$ such that $v < w$ we must have that either
  $v[0 .. r]\in\mathcal{P}$, if $|v|\ge r + 1$, or $v = \cmin^{|v|}$
  otherwise. Since $w[|w|-1] \neq \cmin$, the only position $s$ that
  satisfies $w[s .. |w| - 1] = \cmin^{|w| - s}$ is $|w|$,
  corresponding to the empty word. Hence,
$$
\{ s\ |\ w[s .. |w| - 1] < w \} = \{ s\in Occ_{\mathcal{P}}(w)\ |\ \suf{w}{s} < w \}\cup\{|w|\}
$$
\qed
\end{proof}
Based on these Lemmas, we can devise a faster factorization algorithm
for words containing runs of $\cmin$. The key idea is that, using
Lemma~\ref{lemma:lf-skip}, it is possible to skip symbols in the
computation of $b_1$, if a suitable string matching algorithm is used
to compute $Occ_{\mathcal{P}}(w)$.
W.l.o.g. we assume that the last symbol of $w$ is different from
$\cmin$. In the general case, by Lemma~\ref{lemma:lf-tail}, we can
reduce the factorization of $w$ to the one of its longest prefix with
last symbol different from $\cmin$, as the remaining suffix is a
concatenation of $\cmin$ symbols, whose factorization is a sequence of
factors equal to $\cmin$.
Suppose that $\cmin\cmin$ occurs in $w$. By Lemma~\ref{lemma:lf-split}
we can split the factorization of $w$ in $CFL(u)$ and $CFL(v)$ where
$uv = w$ and $|u| = \min Occ_{\{\cmin\cmin\}}(w)$. The factorization
of $CFL(u)$ can be computed using Duval's original algorithm.

Concerning $v$, let $r = \min \{ i\ |\ v[i]\neq \cmin \}$. By definition $v[0] = v[1] = \cmin$
and $v[|v|-1]\neq \cmin$, and we can apply Lemma~\ref{lemma:lf-skip}
on $v$ to find the ending position $s$ of the first factor in
$CFL(v)$. To this end, we
have to find the position $\min \{i\in Occ_{\mathcal{P}}(v)\ |\
\suf{v}{i} < v\}$, where $\mathcal{P} = \{ \cmin^r c\ |\ c \le v[r]
\}$. For this purpose, we can use any algorithm for multiple string
matching to iteratively compute $Occ_{\mathcal{P}}(v)$ until either a
position $i$ is found that satisfies $\suf{v}{i} < v$ or we reach the
end of the string. Let $h = lcp(v, \suf{v}{i})$ , for a given $i\in
Occ_{\mathcal{P}}(v)$. Observe that $h\ge r$ and, if $v < \suf{v}{i}$,
then, by Lemma~\ref{lemma:lf-lcp-skip}, we do not need to verify the
positions $i'\in Occ_{\mathcal{P}}(v)$ such that $i'\le i + h$.
Given
that all the patterns in $\mathcal{P}$ differ in the last symbol only,
we can express $\mathcal{P}$ more succinctly using a character class
for the last symbol and match this pattern using a string matching
algorithm that supports character classes, such as the algorithms
based on bit-parallelism.
In this respect, SBNDM2 \cite{DurianHPT10}, a variation of the BNDM
algorithm \cite{NavarroR00} is an ideal choice, as
it is sublinear on average. Instead of $\mathcal{P}$, it is naturally
possible to search for $\cmin^r$, but that solution is slower in
practice for small alphabets. Note that the same algorithm can also be
used to compute $\min Occ_{\cmin\cmin} (w)$ in the first phase.

Let $h = lcp(v, v[s .. |v| - 1])$ and $k = 1 + \lfloor h / s\rfloor$.
Based on Lemma~\ref{lemma:lf-lcp}, the algorithm then outputs $v[0 ..
s - 1]$ $k$ times and iteratively applies the above method on $v' =
v[sk .. |v| - 1]$. It is not hard to verify that, if $v'\neq
\varepsilon$, then $|v'|\ge r + 1$, $v'[0 .. r - 1] = \cmin$ and
$v'[|v'| - 1]\neq \cmin$, and so Lemma~\ref{lemma:lf-skip} can be used
on $v'$. The code of the algorithm is shown in
Figure~\ref{fig:duval-variant}. The computation of the value $r' =
\min \{ i\ |\ v'[i]\neq \cmin \}$ for $v'$ takes advantage of the fact
that $v'[0 .. r - 1] = \cmin$, so as to avoid useless comparisons.
If the the total time spent for the iteration over the sets
$Occ_{\mathcal{P}}(v)$ is $O(|w|)$, the full algorithm has also linear
time complexity in the worst case. To see why, it is enough to observe
that the positions $i$ for which the algorithm verifies if $v[i ..
  |v|-1] < v$ are a subset of the positions verified by the original
algorithm.

\setcounter{instr}{0}
\begin{figure}[t]
\begin{center}
\begin{tabular}{|rl|}
\hline
\multicolumn{2}{|l|}{\textsc{LF-rle}$(R)$}\\
\ninstr & $k\leftarrow 0$ \\
\ninstr & \textbf{while} $k < |R|$ \textbf{do} \\
\ninstr & \qquad $(m,q)\leftarrow \textsc{LF-rle-next}(R, k)$ \\
\ninstr & \qquad \textbf{for} $i\leftarrow 1$ \textbf{to} $q$ \textbf{do} \\
\ninstr & \qquad \qquad \textbf{output} $(k, k + m - 1)$ \\
\ninstr & \qquad \qquad $k\leftarrow k + m$ \\
\hline
\end{tabular}
\setcounter{instr}{0}
\begin{tabular}{|rl|}
\hline
\multicolumn{2}{|l|}{\textsc{LF-rle-next}$(R = \langle (c_1, l_1), \ldots, (c_{\rho},l_{\rho}) \rangle, k)$}\\
\ninstr & $i\leftarrow k$ \\
\ninstr & $j\leftarrow k + 1$ \\
\ninstr & \textbf{while} \textsc{true} \textbf{do} \\
\ninstr & \qquad \textbf{if} $i > k$ \textbf{and} $l_{j-1} < l_{i-1}$ \textbf{then} \\
\ninstr & \qquad \qquad $z\leftarrow 1$ \\
\ninstr & \qquad \textbf{else} $z\leftarrow 0$ \\
\ninstr & \qquad $s\leftarrow i - z$ \\
\ninstr & \qquad \textbf{if} $j = |R|$ \textbf{or} $c_j < c_s$ \textbf{or} \\
\ninstr & \qquad \qquad ($c_j = c_s$ \textbf{and} $l_j > l_s$ \textbf{and} $c_j < c_{s+1}$) \textbf{then} \\
\ninstr & \qquad \qquad \textbf{return} $(j - i, \lfloor (j-k-z) / (j-i)\rfloor)$ \\
\ninstr & \qquad \textbf{else} \\
\ninstr & \qquad \qquad \textbf{if} $c_j > c_s$ \textbf{or} $l_j > l_s$ \textbf{then} \\
\ninstr & \qquad \qquad \qquad $i\leftarrow k$ \\
\ninstr & \qquad \qquad \textbf{else} \\
\ninstr & \qquad \qquad \qquad $i\leftarrow i + 1$ \\
\ninstr & \qquad \qquad $j\leftarrow j + 1$ \\
\hline
\end{tabular}
\end{center}
\caption{The algorithm to compute the Lyndon factorization of a run-length encoded string.}
\label{fig:duval-rle}
\end{figure}

\section{Computing the Lyndon factorization of a run-length encoded string}

In this section we present an algorithm to compute the Lyndon
factorization of a string given in RLE form. The algorithm is based on
Duval's original algorithm and on a combinatorial property
between the Lyndon factorization of a string and its RLE, and has
$O(\rho)$-time and $O(1)$-space complexity, where $\rho$ is the length
of the RLE. We start with the following Lemma:

\begin{lemma}
  Let $w$ be a word over $\Sigma$ and let $w_1, w_2, \ldots, w_m$ be
  its Lyndon factorization. For any $1\le i\le |\textsc{rle}(w)|$, let
  $1\le j, k \le m$, $j \le k$, such that $a^{rle}_i\in [a_j, b_j]$
  and $b^{rle}_i\in [a_k, b_k]$. Then, either $j = k$ or $|w_j| =
  |w_k| = 1$.
\end{lemma}
\begin{proof}
  Suppose by contradiction that $j < k$ and either $|w_j| > 1$ or
  $|w_k| > 1$. By definition of $j,k$, we have $w_j \ge w_k$.
  Moreover, since both $[a_j, b_j]$ and $[a_k, b_k]$ overlap with
  $[a^{rle}_i, b^{rle}_i]$, we also have $w_j[|w_j|-1] = w_k[0]$. If
  $|w_j| > 1$, then, by definition of $w_j$, we have $w_j[0] <
  w_j[|w_j|-1] = w_k[0]$. Instead, if $|w_k| > 1$ and $|w_j| = 1$, we
  have that $w_j$ is a prefix of $w_k$. Hence, in both cases we obtain
  $w_j < w_k$, which is a contradiction.\qed
\end{proof}

The consequence of this Lemma is that a run of length $l$ in the RLE
is either contained in \emph{one} factor of the Lyndon factorization,
or it corresponds to $l$ unit-length factors. Formally:

\begin{corollary}
  Let $w$ be a word over $\Sigma$ and let $w_1, w_2, \ldots, w_m$ be
  its Lyndon factorization. Then, for any $1\le i\le
  |\textsc{rle}(w)|$, either there exists $w_j$ such that $[a^{rle}_i,
  b^{rle}_i]$ is contained in $[a_j, b_j]$ or there exist $l_i$
  factors $w_j, w_{j+1}, \ldots, w_{j+l_i-1}$ such that $|w_{j+k}| =
  1$ and $a_{j+k}\in [a^{rle}_i, b^{rle}_i]$, for $0\le k < l_i$.
\end{corollary}

This property can be exploited to obtain an algorithm for the Lyndon
factorization that runs in $O(\rho)$ time. First, we introduce the
following definition:
\begin{definition}
  A word $w$ is a LR word if it is either a Lyndon word or it
  is equal to $a^k$, for some $a\in\Sigma$, $k\ge 2$. The LR
  factorization of a word $w$ is the factorization in LR
  words obtained from the Lyndon factorization of $w$ by merging in a
  single factor the maximal sequences of unit-length factors with the
  same symbol.
\end{definition}

For example, the LR factorization of $cctgccaa$ is $\langle
cctg, cc, aa\rangle$. Observe that this factorization is a
(reversible) encoding of the Lyndon factorization. Moreover, in this
encoding it holds that each run in the RLE is contained in one factor
and thus the size of the LR factorization is $O(\rho)$. Let
$L'$ be the set of LR words. We now present the algorithm
$\textsc{LF-rle-next}(R, k)$ which computes, given an RLE
sequence $R$ and an integer $k$, the longest LR word in $R$
starting at position $k$. Analogously to Duval's algorithm, it reads
the RLE sequence from left to right maintaining two integers, $j$ and
$\ell$, which satisfy the following invariant:
\begin{equation}\label{eq:invariant}
\begin{array}{ll}
c_k^{l_k}\ldots c_{j-1}^{l_{j-1}}\in P'; \\
\ell =
\begin{cases}
|\textsc{rle}(\beta(c_k^{l_k}\ldots c_{j-1}^{l_{j-1}}))| & \text{ if } j - k > 1, \\
0 & \text{otherwise}.
\end{cases}
\end{array}
\end{equation}
The integer $j$, initialized to $k+1$, is the index of the next run to
read and is incremented at each iteration until either $j = |R|$ or
$c_k^{l_k}\ldots c_{j-1}^{l_{j-1}}\notin P'$. The integer $\ell$,
initialized to $0$, is the length in runs of the longest border of
$c_k^{l_k}\ldots c_{j-1}^{l_{j-1}}$, if $c_k^{l_k}\ldots
c_{j-1}^{l_{j-1}}$ spans at least two runs, and equal to $0$
otherwise. For example, in the case of the word $a b^2 a b^2 a b$ we
have $\beta(a b^2 a b^2 a b) = a b^2 a b$ and $\ell = 4$. Let $i = k +
\ell$. In general, if $\ell > 0$, we have
$$
\begin{array}{l}
l_{j-1}\le l_{i-1}, l_{k}\le l_{j-\ell}, \\
\beta(c_k^{l_k}\ldots c_{j-1}^{l_{j-1}}) = c_k^{l_k} c_{k+1}^{l_{k+1}} \ldots c_{i-2}^{l_{i-2}} c_{i-1}^{l_{j-1}} = c_{j-\ell}^{l_k} c_{j-\ell+1}^{l_{j-\ell+1}} \ldots c_{j-2}^{l_{j-2}} c_{j-1}^{l_{j-1}}. \\
\end{array}
$$
Note that the longest border may not fully cover the last (first) run of the corresponding prefix (suffix). Such the case is for example for the word $ab^2 a^2 b$.
However, since $c_k^{l_k}\ldots c_{j-1}^{l_{j-1}}\in P'$ it must hold that $l_{j-\ell} = l_k$, i.e., the first run of the suffix is fully covered.
Let
$$
z =
\begin{cases}
  1 & \text{ if } \ell > 0 \wedge l_{j-1} < l_{i-1}, \\
  0 & otherwise.
\end{cases}
$$
Informally, the integer $z$ is equal to $1$ if the longest border of $c_k^{l_k}\ldots c_{j-1}^{l_{j-1}}$ does not fully cover the run $(c_{i-1},l_{i-1})$.
By~\ref{eq:invariant} we have that $c_k^{l_k} \ldots
c_{j-1}^{l_{j-1}}$ can be written as $(u v)^q u$, where
$$
\begin{array}{l}
q = \lfloor \frac{j-k-z}{j-i}\rfloor, r = z + (j-k-z)\mod (j-i), \\
u = c_{j-r}^{l_{j-r}}\ldots c_{j-1}^{l_{j-1}}, uv = c_k^{l_k}\ldots c_{j-\ell-1}^{l_{j-\ell-1}} = c_{i-r}^{l_{i-r}}\ldots c_{j-r-1}^{l_{j-r-1}}, \\
u v\in L' \\
\end{array}
$$

For example, in the case of the word $a b^2 a b^2 a b$, for $k = 0$, we have $j =
6, i = 4, q = 2, r = 2$.
The algorithm is based on the following Lemma:
\begin{lemma}\label{lemma:lf-rle}
  Let $j, \ell$ be such that invariant~\ref{eq:invariant} holds and
  let $s = i-z$. Then, we have the following cases:
\begin{enumerate}
\item If $c_j < c_s$ then $c_k^{l_k} \ldots c_{j}^{l_j}\notin P'$;
\item If $c_j > c_s$ then $c_k^{l_k} \ldots c_j^{l_j}\in L'$ and~\ref{eq:invariant} holds for $j + 1$, $\ell' = 0$;
\end{enumerate}
Moreover, if $z = 0$, we also have:
\begin{enumerate}
\setcounter{enumi}{2}
\item If $c_j = c_i$ and $l_j \le l_i$,  then $c_k^{l_k} \ldots c_{j}^{l_j}\in P'$ and~\ref{eq:invariant} holds for $j + 1$, $\ell' = \ell + 1$;
\item If $c_j = c_i$ and $l_j > l_i$, either $c_j < c_{i+1}$ and
  $c_k^{l_k} \ldots c_{j}^{l_j}\notin P'$ or $c_j > c_{i+1}$,
  $c_k^{l_k} \ldots c_j^{l_j}\in L'$ and~\ref{eq:invariant} holds for $j + 1$, $\ell' = 0$.
\end{enumerate}
\end{lemma}
\begin{proof}
  The idea is the following: we apply Lemma~\ref{lemma:lf-duval3} with the word $(u v)^q
  u$ as defined above and symbol $c_j$. Observe that $c_j$ is compared
  with symbol $v[0]$, which is equal to $c_{k+r-1} = c_{i-1}$ if $z =
  1$ and to $c_{k+r} = c_i$ otherwise.

  First note that, if $z = 1$, $c_j \neq c_{i-1}$, since otherwise we
  would have $c_{j-1} = c_{i-1} = c_j$. In the first three cases, we
  obtain the first, second and third proposition of
  Lemma~\ref{lemma:lf-duval3}, respectively, for the word $c_k^{l_k}
  \ldots c_{j-1}^{l_{j-1}}c_j$. Independently of the derived
  proposition, it is easy to verify that the same proposition also
  holds for $c_k^{l_k}\ldots c_{j-1}^{l_{j-1}}c_j^m$, $m\le l_j$.
  Consider now the fourth case. By a similar reasoning, we have that
  the third proposition of Lemma~\ref{lemma:lf-duval3} holds for
  $c_k^{l_k}\ldots c_j^{l_i}$. If we then apply
  Lemma~\ref{lemma:lf-duval3} to $c_k^{l_k}\ldots c_j^{l_i}$ and
  $c_j$, $c_j$ is compared to $c_{i+1}$ and we must have $c_j \neq
  c_{i+1}$ as otherwise $c_i = c_j = c_{i+1}$. Hence, either the first
  (if $c_j < c_{i+1}$) or the second (if $c_j > c_{i+1}$) proposition
  of Lemma~\ref{lemma:lf-duval3} must hold for the word
  $c_k^{l_k}\ldots c_j^{l_i+1}$.\qed
\end{proof}

We prove by induction that invariant~\ref{eq:invariant} is maintained.
At the beginning, the variables $j$ and $\ell$ are initialized to
$k+1$ and $0$, respectively, so the base case trivially holds. Suppose
that the invariant holds for $j,\ell$. Then, by
Lemma~\ref{lemma:lf-rle}, either $c_k^{l_k} \ldots c_{j}^{l_{j}}\notin
P'$ or it follows that the invariant also holds for $j+1,\ell'$, where
$\ell'$ is equal to $\ell+1$, if $z = 0$, $c_j = c_i$ and $l_j\le
l_i$, and to $0$ otherwise. When $c_k^{l_k} \ldots c_{j}^{l_{j}}\notin
P'$ the algorithm returns the pair $(j - i, q)$, i.e., the length of
$uv$ and the corresponding exponent. Based on
Lemma~\ref{lemma:lf-duval1}, the factorization of $R$ can then be
computed by iteratively calling \textsc{LF-rle-next}. When a given
call to \textsc{LF-rle-next} returns, the factorization algorithm
outputs the $q$ factors $uv$ starting at positions $k$, $k + (j - i)$,
\ldots, $k + (q - 1) (j - i)$ and restarts the factorization at
position $k + q (j - i)$. The code of the algorithm is shown in
Figure~\ref{fig:duval-rle}. We now prove that the algorithm runs in
$O(\rho)$ time. First, observe that, by definition of LR
factorization, the for loop at line $4$ is executed $O(\rho)$ times.
Suppose that the number of iterations of the while loop at line $2$ is
$n$ and let $k_1, k_2, \ldots, k_{n+1}$ be the corresponding values of
$k$, with $k_1 = 0$ and $k_{n+1} = |R|$. We now show that the $s$-th
call to $\textsc{LF-rle-next}$ performs less than $2(k_{s+1} - k_s)$
iterations, which will yield $O(\rho)$ number of iterations in total.
This analysis is analogous to the one used by Duval. Suppose that
$i'$, $j'$ and $z'$ are the values of $i$, $j$ and $z$ at the end of
the $s$-th call to $\textsc{LF-rle-next}$. The number of iterations
performed during this call is equal to $j' - k_s$. We have $k_{s+1} =
k_s + q(j' - i')$, where $q = \lfloor \frac{j' - k_s - z}{j -
  i'}\rfloor$, which implies $j' - k_s < 2(k_{s+1} - k_s) + 1$, since,
for any positive integers $x, y$, $x < 2\lfloor x/y\rfloor y$ holds.

\section{Experiments with LF-Skip}

The experiments were run on MacBook Pro with the 2.4 GHz Intel Core 2 Duo processor and 2 GB memory. Programs were written in the C programming language and compiled with the gcc compiler (4.8.2) using the -O3 optimization level.

We tested the LF-skip algorithm and Duval's algorithm with various texts. With the protein sequence of the Saccharomyces cerevisiae genome (3 MB), LF-skip gave a speed-up of 3.5 times over Duval's algorithm. Table \ref{table:su} shows the speed-ups for random texts of 5 MB with various alphabets sizes. With longer texts, speed-ups were larger. For example, the speed-up for the 50 MB DNA text (without newlines) from the Pizza\&Chili Corpus\footnote{http://pizzachili.dcc.uchile.cl/} was 14.6 times.

\begin{table}[ht]
\caption{Speed-up of LF-skip with various alphabet sizes in a random text.}
\centering
\begin{tabular}{  c | c   }
\hline\hline 
$|\Sigma|$ & Speed-up \\ 
\hline 
2&9.0\\
3&7.7\\
4&7.2\\
5&6.1\\
6&4.8\\
8&4.3\\
10&3.5\\
12&3.4\\
15&2.4\\
20&2.5\\
25&2.2\\
30&1.9\\
\hline
\end{tabular}
\label{table:su} 
\end{table}

We made also some tests with texts of natural language. Because runs are very short in natural language, the benefit of LF-skip is marginal. We even tried alphabet transformations in order to vary the smallest character of the text, but that did not help.

\section{Conclusions}

In this paper we have presented two variations of Duval's algorithm
for computing the Lyndon factorization of a string. The first
algorithm was designed for the case of small alphabets and is able to
skip a significant portion of the characters of the string for strings
containing runs of the smallest character in the alphabet.
Experimental results show that the algorithm is considerably faster
than Duval's original algorithm. The second algorithm is for strings
compressed with run-length encoding and computes the Lyndon
factorization of a run-length encoded string of length $\rho$ in
$O(\rho)$ time and constant space.

\bibliographystyle{abbrv}
\bibliography{lyndon}

\end{document}